\documentclass[reqno]{amsart}

\usepackage{enumitem}
\usepackage{amssymb}
\usepackage{hyperref}
\usepackage[notcite,notref]{showkeys}

\newcommand*{\mailto}[1]{\href{mailto:#1}{\nolinkurl{#1}}}

\newtheorem{theorem}{Theorem}[section]
\newtheorem{definition}[theorem]{Definition}
\newtheorem{lemma}[theorem]{Lemma}

\newtheorem{proposition}[theorem]{Proposition}

\newcommand{\R}{{\mathbb R}}

 \newcommand{\br}{|\kern-.25em|\kern-.25em|}

\DeclareMathOperator{\re}{Re}

\numberwithin{equation}{section}

\begin{document}

\title[Asymptotic stability of stationary states]{Asymptotic stability of stationary states
 in wave equation coupled to nonrelativistic particle}

\author[E.\ Kopylova]{Elena Kopylova}
\address{Institute for Information Transmission Problems\\ Russian Academy of Sciences}
\email{\mailto{ek@iitp.ru}}

\author[A.\ Komech]{Alexander Komech}
\address{Institute for Information Transmission Problems\\ Russian Academy of Sciences}
\email{\mailto{akomech@iitp.ru}}

\thanks{{\it The research was carried out at the IITP RAS at the expense of the Russian
Foundation for Sciences (project  14-50-00150) }}


\begin{abstract}
We consider the Hamiltonian system consisting      
of scalar wave field  and a single particle coupled in a translation      
invariant manner. The point  particle is subject to an     
external potential. The stationary solutions of the system are      
a Coulomb type wave field centered at those particle positions for which      
the external force vanishes. It is assumed that the charge density satisfies the Wiener
condition which is a version of the ``Fermi Golden Rule''.
We prove that in the large time approximation any finite energy solution,
with the initial state close to the some stable stationary solution, 
is a sum of this stationary solution and a dispersive wave 
which is a solution of the free wave equation.  
\end{abstract}

\maketitle

\section{Introduction} 
Our paper concerns the problem of nonlinear 
field-particle interaction.
We consider 
a scalar real-valued wave field $\phi (x)$ in ${\mathbb R}^3$ 
coupled to a nonrelativistic particle with position $q$ and momentum $p$ 
governed by 
\begin{eqnarray}\label{WP2}
\left\{\begin{array}{rclrcl}
\dot \phi(x,t)\!\!&\!\!=\!\!&\!\!\pi(x,t),&
\dot\pi(x,t)\!\!&\!\!=\!\!&  \Delta\phi(x,t)-\rho(x-q(t)),\medskip\\
\dot q(t)\!\!&\!\!=\!\!&\!\! p(t),&
\dot p(t)\!\!&\!\!=\!\!&\!\!\!-\nabla V(q(t))+\!\!\displaystyle\int \phi(x,t)\nabla\rho(x-q(t))\, dx .
\end{array}\right.
\end{eqnarray} 
This is a Hamilton system with the Hamilton functional 
\begin{equation}\label{hamilq0} 
{\mathcal H}(\phi ,\pi ,q,p)=\frac 12\int\Big(|\pi(x)|^2+ 
|\nabla \phi (x)|^2)dx+\int \phi (x)\rho (x-q)dx+\frac 12 p^2. 
\end{equation} 
The first two equations in \eqref{WP2}  for the fields are equivalent to the 
wave equation with the source $\rho(x-q)$. 
The form of the last two equations
is determined by the choice of the nonrelativistic 
kinetic energy $p^2/2$ in \eqref{hamilq0}.

It is easy to find stationary solutions to the system \eqref{WP2}. We define
For $q\in{\mathbb R}^3$ we set
\begin{equation}\label{sq}
s_{q}(x) =-\int\frac{d^3y}{4\pi|y-x|}\rho(y-q).
\end{equation}
Let $Z=\{q\in{\mathbb R}^3: \nabla V(q)=0\}$ be the set  of critical points for
$V$. Then the set ${\mathcal S}$ of stationary solutions  is given by
\begin{equation}\label{WPss}
{\mathcal S}=\{(\phi, \pi, q,  p)=(s_{q}, 0, q, 0)=:S_{q}|~~q\in Z\}.
\end{equation}
We assume that $V\in C^2({\mathbb R}^3)$ and set
\begin{equation}\label{V-as}
 V_{0}:=\inf_{q\in {\mathbb R}^3}V(q)>-\infty.
\end{equation}
For the charge distribution $\rho$ we assume that
\begin{equation}\label{rho-as}
\rho\in C_0^\infty({\mathbb R}^3),\quad
\rho(x)=0{\rm ~~for~~}|x|\geq R_\rho,~~~~\rho(x)=\rho_r(|x|).
\end{equation}
We also assume that the Wiener condition is satisfied:
\begin{equation}\label{FGR}
\hat\rho(k)= \int d^3x\, e^{ik x}\rho(x)\not=0,\quad k\in{\mathbb R}^3.
\end{equation}
It is an analogue of the Fermi Golden Rule:
the coupling term  $\rho(x-q)$ is not orthogonal
to the eigenfunctions $e^{ikx}$ of the continuous spectrum of the linear part
of the   equation (cf. \cite{SW}).

Finally we assume that some  $q^*\in Z$ is a stable critical point of $V$:
\begin{definition}\label{WP1.1}
 A point $q^*\in Z$ is  stable if $d^2V(q^*)>0$ as a quadratic form.
\end{definition}
Our main results are the following:\\
For solutions to the system (\ref{WP2}) 
with initial data close to  $S_{q^*}=(s_{q^*},0,q^*,0)$ 
we prove the asymptotics
\begin{equation}\label{Yconv}
\Vert\phi(\cdot,t)-s_{q^*}\Vert_{\dot H^1_{-\sigma}}
+\Vert\pi(\cdot,t)\Vert_{L^2_{-\sigma}}+|q(t)-q^*|+|p(t)|={\mathcal O}(t^{-\sigma}),
\quad t\pm\infty,\quad\sigma>1
\end{equation}
in weighted Sobolev norms (see (\ref{Sob})).
Such asymptotics in global energy norm  do not hold
in general since  the field components 
may contain a dispersive term, whose energy
radiates to infinity as $t\to\pm\infty$ but does not converge to zero.
Namely, in global energy norms  we obtain the following scattering asymptotics:
\begin{equation}\label{Si} 
(\phi(x,t), \pi(x,t))\sim (s_{q^*}, 0) +W_0(t)\Phi_\pm,\quad t\to\pm\infty. 
\end{equation} 
Here $W_0(t)$ is the dynamical group of the free wave equation, and
$\Phi_\pm$ are the corresponding asymptotic scattering states, 

Asymptotics similar to (\ref{Yconv}) in local energy semi-norms was  obtained in \cite{KSK} 
in the case of compactly supported difference $\phi(x,0)-s_{q^*}(x)$. We get rid of this
restriction in present paper.

For the proof we establish long-time decay of the linearized  dynamics
using our results \cite{3w} on the dispersion decay for the wave
equation in weighted Sobolev norms. Then we apply the method of majorants.

Let us comment on previous results in these directions.
The asymptotic stability of the solitons was proved in  
\cite{IKV2006} for the system of type (\ref{WP2}) with 
the Klein-Gordon  equation instead of the wave equations.
This result was extended in \cite{IKS2011,KKop2006,KKopS2011,IKV2011}
to similar  system with the Schr\"odinger, Dirac, wave and Maxwell equations. 
The survey of these results can be found in  \cite{Im2013}.

\section{Main results}
To formulate our results precisely, we  introduce 
a suitable phase space.  
Let $L^2$ be the real Hilbert space $L^2({{\mathbb R}}^3)$ with scalar product 
$\langle\cdot,\cdot\rangle$.
Denote $\dot H^1$ the completion of real space $C_0^\infty({\mathbb R}^3)$
with norm $\Vert\nabla\phi(x)\Vert_{L^2}$.
Equivalently, using Sobolev's embedding theorem (see \cite{Li}),
$$
\dot H^1=\{\phi(x)\in L^6({\mathbb R}^3):~~|\nabla\phi(x)|\in L^2\}.
$$
Introduce  the weighted Sobolev spaces $L^2_{\alpha}$ and $\dot H^1_{\alpha}$,
$\alpha\in{\mathbb R}$ with the norms 
\begin{equation}\label{Sob}
\Vert\psi\Vert_{L^2_{\alpha}}:=\Vert(1+|x|)^{\alpha}\psi\Vert_{L^2},\quad
\Vert\psi\Vert_{\dot H^1_{\alpha}}:=\Vert(1+|x|)^{\alpha}\psi\Vert_{\dot H^1}.
\end{equation}

\begin{definition}
 i) The phase space ${\mathcal E}$ is the real Hilbert space
$\dot H^1\oplus L^2\oplus {{\mathbb R}}^3\oplus {{\mathbb R}}^3$ of states
$Y=(\psi ,\pi ,q,p)$ equipped with the finite norm
$$
\Vert Y\Vert_{{\mathcal E}}=\Vert \nabla\psi \Vert_{L^2} +
\Vert\pi \Vert_{L^2}+|q|+|p|.
$$
ii) ${\mathcal E}_{\alpha}$ is the space
$\dot H^1_{\alpha}\oplus L^2_{\alpha}\oplus {{\mathbb R}}^3\oplus {{\mathbb R}}^3$
equipped with the norm
\begin{equation}\label{alfa}
\Vert Y\Vert_{{\mathcal E}_\alpha}=
\Vert \psi \Vert_{\dot H^1_\alpha} +\Vert\pi \Vert_{L^2_\alpha}+|q|+|p|.
\end{equation}
iii) ${\mathcal F}_{\alpha}$ is the space
$\dot H^1_{\alpha}\oplus H^0_{\alpha}$ of fields
$F =(\psi ,\pi )$ equipped with the finite norm
\begin{equation}\label{Falfa}
\Vert \,F\Vert_{{\mathcal F}_\alpha}=\Vert \psi \Vert_{\dot H^1_\alpha} +\Vert\pi \Vert_{L^2_\alpha}.
\end{equation}
\end{definition}
We consider the Cauchy problem for the Hamiltonian system (\ref{WP2})
\begin{equation}\label{WP2.1}
\dot Y(t)=F(Y(t)),~~t\in{\mathbb R},~~Y(0)=Y_0.
\end{equation}
All derivatives are understood in the sense of distributions. Here,
$$
Y(t)=(\phi(t), \pi(t), q(t), p(t)),~~ Y_0=(\phi_0, \pi_0, q_0, p_0)\in {\mathcal E}
$$.
\begin{lemma}\label{WPexistence} (sf. \cite [Lemma 2.1]{KSK})
Let  (\ref{V-as}) and (\ref{rho-as}) be satisfied.
 Then the following assertions hold.\\
 (i) For every $Y_0\in {\mathcal E}$ the Cauchy problem (\ref{WP2.1}) has a unique
 solution $Y(t)\in C({\mathbb R}, {\mathcal E})$.\\
 (ii) For every $t\in{\mathbb R}$ the map $Y_0\mapsto Y(t)$ is continuous on ${\mathcal E}$.\\
 (iii) The energy is conserved, i.e.,
 \begin{equation}\label{WPEC}
  {\mathcal H}(Y(t))={\mathcal H}(Y_0)~~~for~~t\in{\mathbb R}.
 \end{equation} 
  (iv) The energy is bounded from below, and
 \begin{equation}\label{WPHm}
  \inf_{Y\in {\mathcal E}} {\mathcal H}(Y)=V_{0}+\frac 12(\rho,\Delta^{-1}\rho)
\end{equation}
\end{lemma}
\medskip
Our first result is the following long-time  convergence 
in ${\mathcal E}_{-\sigma}$ to the stationary stable state:
\begin{theorem}\label{WPB}
Let conditions (\ref{V-as})- (\ref{FGR} hold, and let
$Y(t)$ be a solution to the Cauchy problem  (\ref{WP2.1}) with 
initial state $Y_0\in{\mathcal E}$  close 
to $S_{q^*}=(s_{q^*},0,q^*,0)$ with stable $q^*\in Z$:
\begin{equation}\label{close}
d_0:=\Vert \nabla(\phi_0-s_{q^*})\Vert_{L^2_\sigma}+\Vert \pi_0\Vert_{L^2_\sigma}
+|q_0-q^*|+|p_0|\ll 1,
\end{equation}
where $\sigma> 1$.
Then for sufficiently small  $d_0$ 
\begin{equation}\label{WPEE1}
\Vert Y(t)-S_{q^*}\Vert_{{\mathcal E}_{-\sigma}}\le C(d_0)(1+|t|)^{-\sigma},\quad t\in\R.
\end{equation}
\end{theorem}
Our second result is the following scattering long-time asymptotics 
in global energy norms for the field components of the solution:
\begin{theorem}\label{main}
Let the assumptions of Theorem \ref{WPB} hold.
Then for  sufficiently small $d_0$
\begin{equation}\label{S}
(\phi(x,t), \pi(x,t))= (s_{q^*}, 0) +W_0(t)\Phi_\pm+r_\pm(x,t), \quad t\to\pm\infty,
\end{equation}
where $W_0(t)$ is the dynamical group of the
free wave equation, $\Phi_\pm\in \dot H^1\oplus L^2$,  and
\begin{equation}\label{rm} 
\Vert r_\pm(t)\Vert_{\dot H^1\oplus L^2}={\mathcal O}(|t|^{-\sigma+1}),\quad t\to\pm\infty. 
\end{equation}
\end{theorem}
It suffices to  prove  (\ref{S}) for positive $t\to+\infty$ 
since the system (\ref{WP2}) is time reversible.  
 \setcounter{equation}{0}
 \section{Linearization at a stationary state}

For notational simplicity we also assume isotropy in the sense that
\begin{equation}\label{WPhesse}
 \partial_i\partial_j V(q^*)=\omega^2_0\delta_{ij},\
    i,j=1,2,3,\ \omega_0>0\,.
\end{equation}
Without loss of generality we take $q^*=0$.\\
Let $S_q = S_0 = (s_0, 0, 0, 0)$ be the
stationary state of (\ref{WP2}) corresponding to $q^*=0$, and
$Y_0=(\phi_0, \pi_0, q_0, p_0)\in {\mathcal E}$ be an arbitrary
initial data  satisfying (\ref{close}). 
Consider  $Y(x,t)=(\phi(x,t),\pi(x,t),q(t), p(t))
\in{\mathcal E}$ the solution to (\ref{WP2}) with $Y(0)=Y_0$.

To linearize (\ref{WP2}) at $S_0$, we
set $\phi(x,t)=\psi(x,t)+s_0(x)$. Then (\ref{WP2}) becomes
\begin{eqnarray}\label{WPnonlin}
\left\{\begin{array}{lll}
 \dot{\psi}(x, t)&=&\pi (x, t)\\
 \dot{\pi}(x, t)&=&\Delta \psi(x, t)+\rho(x)-\rho(x-q(t))\\
 \dot q(t) &=&  p(t)\\
 \dot{p}(t)&=& -\nabla V(q(t))+
 \displaystyle\int d^{3}x\,\psi(x, t)\,\nabla\rho(x-q(t))\\
&&+\displaystyle\int d^{3}x\,s_0(x)[\nabla\rho (x-q(t))-\nabla\rho(x)]
\end{array}\right.
\end{eqnarray}
Introducing $X(t)=Y(t)-S_0=(\psi(t),\pi(t),q(t),p(t))\in C({\mathbb R}, {\mathcal E})$,
we rewrite the nonlinear system (\ref{WPnonlin}) in the form
\begin{equation}\label{WPAB}
 \dot X(t)=AX(t)+B(X(t)).
\end{equation}
Here $A$ is the linear operator defined by
\begin{equation}\label{AA} 
A\left( \begin{array}{c} 
\psi \\ \pi \\ q \\ p 
\end{array} \right):=\left( 
\begin{array}{cccc} 
0 & 1 & 0 & 0 \\ 
\Delta  & 0 & \nabla\rho\cdot & 0 \\ 
0 & 0 & 0 & E \\ 
\langle\cdot,\nabla\rho\rangle & 0 & -\omega^2_0-\omega^2_1 & 0 
\end{array} 
\right)\left( 
\begin{array}{c} 
\psi \\ \pi \\ q \\ p 
\end{array} 
\right)
\end{equation} 
with
\begin{equation}\label{WPO1}
 \omega^2_1\delta_{ij}=\frac{1}{3}\,\Vert \rho\Vert^2_{L^2}\delta_{ij}=-\,
  \int d^{\,3}x\,\partial_i s_0(x)\partial_j\rho(x) \,.
 \end{equation}
Here, the factor
$1/3$ is due to a spherical symmetry of $\rho(x)$ (sf. (\ref{rho-as})).

The nonlinear part is given by
\begin{equation}\label{WPBB}
B(X)=(0,\pi_1,0,p_1),
\end{equation}
where
\begin{equation}\label{pi1}
\pi_1=\rho(x)-\rho(x-q)-\nabla\rho(x)\cdot q
\end{equation}
and
\begin{eqnarray}\nonumber
p_1&=&-\nabla V(q)+\omega_0^2 q
+\int d^{\,3}x\,\psi(x) [\nabla\rho(x-q)-\nabla\rho (x)]\\
\label{p1}
&+&\int d^{\,3}x\,\nabla s_0(x)[\rho(x)-\rho(x-q)-\nabla\rho (x)\cdot q].
\end{eqnarray}
Let us consider the Cauchy problem for the linear equation
\begin{equation}\label{WPlin}
\dot Z(t)=AZ(t),\quad Z=(\Psi,\Pi, Q, P),\quad t\in{\mathbb R},
\end{equation}
with initial condition
\begin{equation}\label{WPic}
Z|_{t=0}=Z_0.
\end{equation}
System (\ref{WPlin}) is a formal Hamiltonian system with the quadratic Hamiltonian
$$
{\mathcal H}_0(Z)=\frac 12 \Big( {P^2} +\omega^2 Q^2+
\int d^3x\,(|\Pi(x)|^2+|\nabla\Psi(x)|^2-2\Psi(x)\nabla\rho(x)\cdot Q) \Big),
$$
which is the formal Taylor expansion of ${\mathcal H}(Y_0+Z)$ up to second order at $Z=0$.
 \begin{lemma}\label{WPexlin}
 Let the condition (\ref{rho-as}) be satisfied. Then the following assertion hold\\
 (i) For every $Z_0\in{\mathcal E}$ the Cauchy problem (\ref{WPlin}), (\ref{WPic})
 has a unique solution $Z(\cdot)\in C({\mathbb R},{\mathcal E})$.\\
 (ii) For every $t\in{\mathbb R}$ the map $U(t):Z_0\mapsto Z(t)$ is continuous on ${\mathcal E}$.\\
 (iii) For $Z_0\in{\mathcal E}$ the energy ${\mathcal H}_0$ is finite and conserved, i.e.
 \begin{equation}\label{WPec}
 {\mathcal H}_0(Z(t))={\mathcal H}_0(Z_0)~~for~~t\in{\mathbb R}.
 \end{equation}
 iv) For $Z_0\in{\mathcal E}$
 \begin{equation}\label{WPlinb}
 \Vert Z(t)\Vert_{\mathcal E} \leq C~~~for~~t\in{\mathbb R}
 \end{equation}
 with $C$ depends only on the norm $\Vert Z_0\Vert_{\mathcal E}$.
 \end{lemma}
\setcounter{equation}{0}
\section{Decay of linearized dynamics}
We  prove the following long-time decay of the solution $Z(t)$ to (\ref{WPlin}):
\begin{proposition}\label{TDL}
Let the conditions (\ref{rho-as}) and (\ref{FGR}) hold, and let $Z_0\in{\mathcal E}$ be such that 
$$
\Vert\nabla\Psi_0\Vert_{L^2_\sigma}+\Vert\Pi_0\Vert_{L^2_\sigma}<\infty
$$ 
with some $\sigma>1$.  Then for $Z(t)=U(t)Z_0$ 
\begin{equation}\label{Z-dec}
\Vert Z(t)\Vert_{{\mathcal E}_{-\sigma}}\le C(\rho,\sigma)(1+|t|)^{-\sigma+1}
(\Vert\nabla\Psi_0\Vert_{L^2_\sigma}+\Vert\Pi_0\Vert_{L^2_\sigma}).
\end{equation}
\end{proposition}
\bigskip
To prove this assertion  we apply the Fourier-Laplace transform
\begin{equation}\label{FL}
\tilde Z(\lambda)=\Lambda Z(t)=\int_0^\infty e^{-\lambda t}Z(t)dt,\quad\re\lambda>0
\end{equation}
to (\ref{WPlin}). We expect that
the solution $Z(t)$ is bounded in the norm $\Vert\cdot\Vert_{{\mathcal E}}$.
Then the integral (\ref{FL}) converges and is analytic for $\re\lambda>0$,
and
\begin{equation}\label{PW}
\Vert\tilde Z(\lambda)\Vert_{{\mathcal E}}\le \displaystyle\frac{C}{\re\lambda},\quad \re\lambda>0.
\end{equation}
Applying the Fourier-Laplace transform to (\ref{WPlin}), we obtain that
\begin{equation}\label{FLA}
\lambda\tilde Z(\lambda)=A\tilde Z(\lambda)+Z_0,\quad \re\lambda>0.
\end{equation}
Hence the solution $Z(t)$ is given by
\begin{equation}\label{FLAs}
\tilde Z(\lambda)=-(A-\lambda)^{-1}Z_0,\quad \re\lambda>0
\end{equation}
By (\ref{PW}), the  resolvent $R(\lambda)=(A-\lambda)^{-1}$ exists and is analytic in ${\mathcal E}$
for $\re\lambda>0$.

Let us  construct the resolvent for $\re\lambda>0$.
Equation (\ref{FLA}) takes the form
\begin{equation}\label{eq1}
\lambda\left(
\begin{array}{c}
\tilde\Psi \\ \tilde\Pi \\ \tilde Q \\ \tilde P
\end{array}\right)=\left(
\begin{array}{r}
\tilde\Pi \\
\Delta\tilde\Psi+\tilde Q\cdot\nabla\rho \\
\tilde P \\
-\langle\nabla\tilde\Psi,\rho\rangle-\omega^2\tilde Q
\end{array}
\right)+\left(
\begin{array}{c}
\Psi_0 \\ \Pi_0 \\ Q_0 \\ P_0
\end{array}
\right)
\end{equation}
where  $\omega^2=\omega_0^2+\omega_1^2$.
\smallskip\\
{\it Step i)} We consider  the first two equations of (\ref{eq1}):
\begin{equation}\label{F1}
\left\{
\begin{array}{rcl}
-\lambda\tilde\Psi+\tilde\Pi&=&-\Psi_0
\\
\Delta\tilde\Psi-\lambda\tilde\Pi&=&-\Pi_0-\tilde Q\cdot\nabla\rho
\end{array}\right.
\end{equation}
A solutions to the system (\ref{F1}) admits  the convolution representation
\begin{equation}\label{Psi}
\left\{\begin{array}{rcl}
\tilde\Psi&=&\lambda g_{\lambda}*\Psi_0+g_{\lambda}*\Pi_0+(g_{\lambda}*\nabla\rho)\cdot\tilde Q
\\
\tilde\Pi&=&\Delta g_{\lambda}*\Psi_0+\lambda g_{\lambda}*\Pi_0+\lambda (g_{\lambda}*\nabla\rho)\cdot\tilde Q,
\end{array}\right.
\end{equation}
where 
\begin{equation}\label{dete}
g_\lambda(z)=(-\Delta+\lambda^2)^{-1}=\frac{e^{-\lambda|z|}}{4\pi|z|}.
\end{equation}
\noindent{\it Step ii)}
We consider  the last two equations of (\ref{eq1}):
\begin{equation}\label{lte}
\left\{\begin{array}{rcl}
-\lambda\tilde Q+\tilde P&=&- Q_0\\
-\omega_2 \tilde Q-\langle\nabla\tilde\Psi,\rho\rangle-\lambda\tilde P&=&- P_0
\end{array}\right.
\end{equation}
Let us write the first equation of (\ref{Psi}) in the form
$\tilde\Psi(x)=\tilde\Psi_1(\tilde Q)+\tilde\Psi_2(\Psi_0,\Pi_0)$,
where
\begin{equation}\label{Psi12}
\tilde\Psi_1(\tilde Q)
=\tilde Q\cdot(g_{\lambda}*\nabla\rho),
~~~~~~~~~
\tilde\Psi_2(\Psi_0,\Pi_0)=\lambda g_{\lambda}*\Psi_0+g_{\lambda}*\Pi_0.
\end{equation}
Then the second equation in (\ref{lte}) becomes
$$
-\omega^2 \tilde Q-\langle\nabla\tilde\Psi_1,
\rho\rangle-\lambda\tilde P=-P_0+\langle\nabla\tilde\Psi_2,\rho\rangle=:- P_0'.
$$
Now we compute  term $\langle\nabla\tilde\Psi_1,\rho\rangle$:
\begin{eqnarray}\nonumber
\langle\nabla\tilde\Psi_1,\rho\rangle&=&
-\langle\tilde\Psi_1,\partial_i\rho\rangle=
-\langle\sum\limits_j(g_{\lambda}*\partial_j\rho)\tilde Q_j,\partial_i\rho\rangle\\
\nonumber
&=& -\sum\limits_j\langle g_{\lambda}*\partial_j\rho,\partial_i\rho\rangle\tilde Q_j=
-\sum\limits_j H_{ij}(\lambda)\tilde Q_j,
\end{eqnarray}
where
\begin{eqnarray}\label{Cij}
H_{ij}(\lambda):&=&\langle g_{\lambda}*\partial_j\rho,
\partial_i\rho\rangle=
\langle i\hat g_{\lambda}(k)k_j\hat\rho(k),ik_i\hat\rho(k)\rangle
\nonumber\\
\nonumber\\
&=&\langle\frac{ik_j\hat\rho(k)}{k^2+\lambda^2},ik_i\hat\rho(k)\rangle=
\int\frac{k_ik_j|\hat\rho(k)|^2dk}{k^2+\lambda^2}.
\end{eqnarray}
The matrix $H$ with entries $H_{jj}$, $1\le j\le3$,
is well defined for $\re\lambda>0$ since the denominator does not vanish.
The matrix $H$ is diagonal; moreover, 
\begin{equation}\label{hii}
H_{11}(\lambda)=H_{22}(\lambda)=H_{33}(\lambda)=h(\lambda).
\end{equation}
Finally, the system (\ref{lte}) takes the form
\begin{equation}\label{Mlam}
M(\lambda)\left(
\begin{array}{c}
\tilde Q \\ \tilde P
\end{array}
\right)=\left(
\begin{array}{c}
 Q_0 \\  P_0'
\end{array}\right),
\end{equation}
where
\[
M(\lambda)=\left(
\begin{array}{cc}
\lambda E  & -E \\
\omega^2 E-H(\lambda) & \lambda E
\end{array}\right).
\]
\begin{lemma}\label{cmf}
The matrix-valued function $M(\lambda)$ ( $M^{-1}(\lambda)$)
admits an analytic (meromorphic) continuation to the entire complex plane ${\mathbb C}$.
\end{lemma}
\begin{proof}
The Green function $g_{\lambda}$ admits an analytic continuation in $\lambda$
to the entire complex plane ${\mathbb C}$.
Then an analytic continuation of  $M(\lambda)$ exists in view of
(\ref{Cij}) since the function $\rho(x)$
is compactly supported because of  (\ref{rho-as}).
Then the inverse matrix is meromorphic since it exists
for large $\re\lambda$. This fact follows from (\ref{Mlam})
since $H(\lambda)\to 0$, $\re\lambda\to\infty$ in view of (\ref{Cij}).
\end{proof}
Since the  matrix $H(\lambda)$ is  diagonal, the matrix 
$M(\lambda)$ is equivalent to three independent $2\times2$- matrices. 
Namely, let us transpose the columns and rows of the matrix $M(\lambda)$ 
in the order $(142536)$. Then we get the matrix with three $2\times 2$- blocks on the main diagonal. 
Therefore, the determinant of $M(\lambda)$ is the product of  determinants 
of the three matrices. Namely, 
\begin{equation}\label{detM}
\det M(\lambda)=(\lambda^2+\omega^2-h(\lambda))^3
=(\lambda^2+\omega_0^2+\omega_1^2-h(\lambda))^3,
\end{equation}
where
$$
\omega_1^2=\int\frac{k_1^2|\hat\rho(k)|^2dk}{k^2},
\quad h(\lambda)=\int\frac{k_1^2|\hat\rho(k)|^2dk}{k^2+\lambda^2}.
$$
\begin{proposition}
The matrix $M^{-1}(i\nu+0)$ is analytic in $\nu\in{\mathbb R}$.
\end{proposition}
\begin{proof} It suffices to prove that the limit matrix $M(i\nu+0)$ 
is invertible for $\nu\in{\mathbb R}$ if
$\rho$ satisfies the Wiener condition (\ref{FGR}).
Formula (\ref{detM}) implies $\det M(0)=\omega_0^2>0$. For $\nu\not=0$, $\nu\in{\mathbb R}$,
we consider
\begin{equation}\label{Hjjlim}
h(i\nu+\varepsilon)=\int\frac{k_1^2|\hat\rho(k)|^2dk}{k^2-(\nu-i\varepsilon)^2},
\quad\varepsilon>0.
\end{equation}
The denominator $D(\nu,k)=k^2-\nu^2$ vanishes
on  $T_{\nu}=\{k:k^2=\nu^2\}$. Denote by $dS$  the surface area element.
Then from the Sokhotsky-Plemelj formula for $C^1$-functions it follows that
\begin{equation}\label{ImHjj}
\Im h(i\nu+0)=-\frac{\nu}{|\nu|}\pi\int_{T_{\nu}}
\frac{k_1^2|\hat\rho(k)|^2}{|\nabla D(\nu,k)|}dS\not =0,
\end{equation}
since the integrand in (\ref{ImHjj})
is positive by the Wiener condition (\ref{FGR}).
Now, the invertibility  of $M(i\nu)$ follows from (\ref{detM}).
\end{proof}
\subsection{Time decay }
Here we prove Proposition \ref{TDL}.
First, we obtain  decay (\ref{Z-dec}) for the vector components $Q(t)$ and $P(t)$ of $Z(t)$.
By (\ref{Mlam}), the components are given by the
Fourier integral
\begin{equation}\label{QP1i}
\left(\!\!\begin{array}{c}
Q(t) \\ P(t)
\end{array}\!\!\right)=\displaystyle\frac 1{2\pi}
\int e^{i\nu t}M^{-1}(i\nu\!+0)\left(\!\!\begin{array}{c}
Q_0 \\ P_0'
\end{array}\!\!\right)d\nu={\mathcal L}(t)\left(\!\!\begin{array}{c}
Q_0 \\ P_0
\end{array}\!\!\right)+{\mathcal L}(t)*\left(\!\!\begin{array}{c}
0 \\ f(t)
\end{array}\!\!\right),
\end{equation}
where
$$
{\mathcal L}(t)=\displaystyle\frac 1{2\pi}\int e^{i\nu t}M^{-1}(i\nu+0)d\nu=\lambda^{-1}M^{-1}(i\nu+0),\quad
$$
and
\begin{eqnarray}\nonumber
f(t)&=&\Lambda^{-1}[\langle\Psi_2(\Psi_0,\Pi_0),\nabla\rho\rangle]
=\Lambda^{-1}[\langle i\nu g_{i\nu}*\Psi_0+g_{i\nu}*\Pi_0,\nabla\rho\rangle]\\
\label{f-def}
&=&\langle W_0(t)[(\Psi_0,\Pi_0)],\nabla\rho\rangle.
\end{eqnarray}
\\
We write the nonzero entries  of the matrix $M(i\nu+0)$: 
$$
\frac{i\nu}{-\nu^2+\omega^2-h(i\nu+0)},\quad\frac{1}{-\nu^2+\omega^2-h(i\nu+0)},
\quad\frac{-\omega^2+h(i\nu)}{-\nu^2+\omega^2-h(i\nu+0)}.
$$
Hence, 
$$
|M^{-1}(i\nu+0)|\leq \frac{C}{|\nu |},\quad
|\partial^k M^{-1}(i\nu+0)|\leq \frac{C_k}{|\nu |^2},
\quad\nu \in {\mathbb R},\quad |\nu|\ge 1,\quad k\in{\mathbb N}.
$$
Therefore, ${\mathcal L}(t)$ is continuous in $t\in{\mathbb R}$ and
\begin{equation}\label{cM-dec}
{\mathcal L}(t)={\mathcal O}(|t|^{-N}),\quad t\to\infty,\quad\forall N>0.
\end{equation}
For the solutions of the free wave equation
the following dispersion decay holds:
\begin{lemma}\label{WPI0} (sf. \cite[Proposition 2.1]{3w})
Let $(\Psi_0,\Pi_0)\in{\mathcal F}_0$ be such that 
\[
\Vert (\Psi_0,\Pi_0)\Vert_{\mathcal F_\sigma}<\infty
\]
with some $\sigma >1$. Then
\begin{equation}\label{lins}
  \Vert  W(t)[(\Psi_0,\Pi_0)]\Vert_{{\mathcal F}_{-\sigma}}
\le C(1+|t|)^{-\sigma}\Vert(\Psi_0,\Pi_0)\Vert_{\mathcal F_\sigma}, \quad t\in{\mathbb R}.
\end{equation}
\end{lemma}
\bigskip
Lemma \ref{WPI0} and the definition (\ref{f-def}) imply 
\begin{equation}\label{f-dec}
|f(t)|\le C(\sigma,\rho)(1+|t|)^{-\sigma}\Vert(\Psi_0,\Pi_0)\Vert_{\mathcal F_\sigma},\quad t\in{\mathbb R}.
\end{equation}
Therefore, (\ref{QP1i}), (\ref{cM-dec}) and (\ref{f-dec}) imply
\begin{equation}\label{QP}
|Q(t)|+|P(t)|\le C(\sigma,\rho)(1+|t|)^{-\sigma}\Vert(\Psi_0,\Pi_0)\Vert_{\mathcal F_\sigma}
\end{equation}
Then  (\ref{Z-dec})  holds for the vector components $Q(t)$ and $P(t)$.
\medskip\\
Now we  prove  (\ref{Z-dec}) for the field components of $Z(t)$.
The first two equations of (\ref{WPlin}) have the form
\begin{equation}\label{linf}
\left(
\begin{array}{ll}
\dot \Psi(t)\\ \dot \Pi(t) 
\end{array}
\right)
=\left(
\begin{array}{ll}
0 & 1 \\
\Delta & 0
\end{array}
\right)\left(
\begin{array}{ll}
\Psi(t)\\ \Pi(t) 
\end{array}
\right)+\left(
\begin{array}{l}
0 \\
Q(t)\cdot\nabla\rho
\end{array}
\right).
\end{equation}
The integrated version of (\ref{linf}) reads
\begin{equation}\label{Duh}
(\Psi(t),\Pi(t))=W(t)[(\Psi_0,\Pi_0)]+\int_0^tW(t-s)
[0, Q(s)\cdot\nabla\rho]ds,~~~~~~t\ge 0.
\end{equation}
By (\ref{lins}) and (\ref{QP}), 
$$
\Vert(\Psi(t),\Pi(t))\Vert_{{\mathcal F}_{-\sigma}}\le C(\rho,\sigma)(1+|t|)^{-\sigma+1}
(\Vert\nabla\Psi_0\Vert_{L^2_\sigma}+\Vert\Pi_0\Vert_{L^2_\sigma}),~~~~t\in{\mathbb R}.
$$
Proposition \ref{TDL} is proved.
\hfill$\Box$

 \setcounter{equation}{0}
 \section{Asymptotic stability of stationary states}
Here we prove Theorem \ref{WPB}.
First, we obtain bounds for the nonlinear part $B(X(t))=(0,\pi_1,0,p_1)$ 
defined in (\ref{WPBB}) - (\ref{p1}). We have
\begin{eqnarray}\nonumber
\Vert\pi_1(t)\Vert_{L^2_{\sigma}}&\le& {\mathcal R}(|q|)|q(t)|^2,\\
\nonumber
|p_1(t)|&\le&{\mathcal R}(|q|)[|q(t)|\Vert\psi(t)\Vert_{\dot H^1_{-\sigma}}+|q(t)|^2],
\end{eqnarray}
where $R(A)$ denotes a positive function that remains bounded for sufficiently small $A$.
Hence
\begin{equation}\label{B-est}
\Vert B(t)\Vert_{{\mathcal E}_\sigma}\le {\mathcal R}(|q|)\Vert X(t)\Vert_{{\mathcal E}_{-\sigma}}^2.
\end{equation}
We introduce the  majorant
\begin{equation}\label{maj}
m(t)=\sup_{0\leq s\leq t}(1+s)^{-\sigma}\Vert X(s)\Vert_{{\mathcal E}_{-\sigma}}.
\end{equation}
We fix $\varepsilon>\Vert X(0)\Vert_{{\mathcal E}_{-\sigma}}$ and introduce the
{\it exit time}
\begin{equation}\label{t*}
t_*=\sup \{t>0:m(t)\le \varepsilon\}.
\end{equation}
The integrated version of (\ref{WPAB}) is written as
\begin{equation}\label{WPkolb}
 X(t)=e^{At}X(0)+\int^t_0 ds\,e^{A(t-s)} B(X(s)).
\end{equation}
Proposition \ref{Z-dec}  implies the integral inequality
\begin{eqnarray}\label{WPiin}
  \Vert X(t)\Vert_{{\mathcal E}_{-\sigma}}&\le& C\,
   (1+|t|)^{-\sigma}\Vert(\psi_0,\pi_0)\Vert_{\mathcal F_\sigma}\\
\nonumber
   &+& C \int^t_0 ds \, (1+|t-s|)^{-\sigma}{\|X(s)\|}^2_{{\mathcal E}_{-\sigma}},
\end{eqnarray}
for $t< t_*$.
We multiply both sides of (\ref{WPkolb}) by $(1+t)^{-\sigma}$,
and take the supremum in $t\in[0,t_*]$. Then
$$
m(t)\le C\Vert(\psi_0,\pi_0)\Vert_{\mathcal F_\sigma}+C\sup_{t\in[0,t_*]}\displaystyle
\int_0^t\frac{(1+t)^{\sigma}}{(1+|t-s|)^{\sigma}}\frac{m^2(s)}
{(1+s)^{2\sigma}}ds
$$
for $t< t_*$. Since $m(t)$ is a monotone increasing function, we get
\begin{equation}\label{mest}
m(t)\le C(\Vert(\psi_0,\pi_0)\Vert_{\mathcal F_\sigma}+Cm^2(t)I(t),\quad
t\le t_*,
\end{equation}
where
$$
I(t)=\int_0^t\frac{(1+t)^{\sigma}}{(1+|t-s|)^{\sigma}}\frac{m^2(s)}
{(1+s)^{2\sigma}}ds\le\overline I<\infty,\quad t\ge 0,\quad\sigma>1.
$$
Therefore, (\ref{mest}) takes the form
\begin{equation}\label{mest1}
m(t)\le C\Vert(\psi_0,\pi_0)\Vert_{\mathcal F_\sigma}+C\overline I m^2(t),
\quad t\le t_*,
\end{equation}
which implies that $m(t)$ is bounded for $t<t_*$; moreover,
\begin{equation}\label{m2est}
m(t)\le C\Vert(\psi_0,\pi_0)\Vert_{\mathcal F_\sigma},\quad t<t_*,
\end{equation}
since $m(0)=\Vert(\psi_0,\pi_0)\Vert_{\mathcal F_\sigma}$
is sufficiently small by (\ref{close}).

The constant $C$ in the estimate (\ref{m2est}) is independent of $t_*$.
We choose $d_{0}$ in (\ref{close}) so small that
$\Vert(\psi_0,\pi_0)\Vert_{\mathcal F_\sigma}<\varepsilon/(2C)$, 
which is possible due to (\ref{close}).
Then the estimate (\ref{m2est}) implies 
$t_*=\infty$, and (\ref{m2est}) holds for all $t>0$ if $d_{0}$ is small enough.
 {\hfill $\Box$}
\section{Scattering asymptotics}
\setcounter{equation}{0}
Here we prove Theorem \ref{main}. From the first two equations of (\ref{WP2})
we obtain the inhomogeneous wave equation for the difference 
$F(x,t)=(\psi(x,t),\pi(x,t))=(\phi(x,t),\pi(x,t))-(s_{q^*},0)$:
\begin{eqnarray}
\dot \psi(x,t)&=&\pi(x,t),\\
\nonumber
\dot \pi(x,t)&=&\Delta \psi(x,t)+\rho(x-q^*)-\rho(x-q(t)).
\end{eqnarray} 
Then 
\begin{equation}\label{eqacc} 
F(t)=W_0(t)F(0)-\int_0^tW_0(t-s)[(0,\rho(x-q^*)-\rho(x-q(s)))]ds. 
\end{equation} 
To obtain  the asymptotics (\ref{S}), it suffices to prove that 
$F(t)=W_0(t)\Phi_++r_+(t)$ for some $\Phi_+\in\dot H^1\oplus L^2$ and 
$\Vert r_+(t)\Vert_{\dot H^1\oplus L^2}={\mathcal O}(t^{-\sigma+1})$. 
This fact is equivalent to the asymptotics 
\begin{equation}\label{Sme} 
W_0(-t)F(t)=\Phi_++r_+'(t),~~~~~ 
\Vert r_+'(t)\Vert_{\\dot H^1\oplus L^2}={\mathcal O}(t^{-\sigma+1}), 
\end{equation} 
since $W_0(t)$ is a unitary group on  $\dot H^1\oplus L^2$ by the energy 
conservation for the free wave equation. Finally, the
asymptotics (\ref{Sme}) hold since (\ref{eqacc}) implies  
\begin{equation}\label{duhs} 
W_0(-t)F(t)= 
F(0)-\int_0^t W_0(-s)R(s)ds,\quad R(s)=(0,\rho(x-q^*)-\rho(x-q(s))  
\end{equation} 
We set
\begin{equation}\label{Pr}
\Phi_+=F(0)-\int_0^\infty W_0(-s)R(s)ds,
\end{equation}
and
$$
r_+'(t)=\int_t^\infty W_0(-s)R(s)ds
$$
The integral  on the right hand side of (\ref{Pr}) 
converges in  $\dot H^1\oplus L^2$ 
with the rate ${\mathcal O}(t^{-\sigma+1})$ because 
$\Vert  W_0(-s)R(s)\Vert_{\dot H^1\oplus L^2} ={\mathcal O}(s^{-\sigma})$ by the unitarity 
of $W_0(-s)$ and the decay rate 
$\Vert R(s)\Vert_{\dot H^1\oplus L^2}={\mathcal O}(s^{-\sigma})$ 
which follows  
from  the conditions (\ref{rho-as}) on $\rho$ 
and the asymptotics (\ref{WPEE1}).
Hence, $\Phi_+\in\dot H^1\oplus L^2$ and  (\ref{rm}) holds.
 {\hfill $\Box$}


\begin{thebibliography}{99}


\bibitem{Im2013} 
V. Imaikin, 
Soliton asymptotics for systems of `field-particle' type, 
\textit{Russian Math. Surveys} \textbf{68} (2013), no. 2, 227--281.
(Translation from  
\textit{Uspekhi Mat. Nauk} \textbf{68} (2013), no. 2(410), 33--90. 
[Russian])

\bibitem{IKV2006} 
V. Imaykin, A. Komech, B. Vainberg, 
On scattering of solitons for 
the Klein--\allowbreak Gordon equation coupled to a~particle, 
\textit{Comm. Math. Phys.} \textbf{268} (2006), no. 2, 321--367. 
arXiv:math.AP/0609205 
 
\bibitem{IKS2011} 
V. Imaykin, A. Komech, H. Spohn, 
Scattering asymptotics for a~charged particle 
coupled to the Maxwell field, 
\textit{J. Math. Physics} \textbf{52} 
(2011), no. 4, 042701-1--042701-33. 
arXiv:0807.1972 
 
\bibitem{KKop2006} 
 A. Komech, E. Kopylova, 
 Scattering of solitons for Schr\"odinger equation coupled to a~particle, 
\textit{Russian J. Math. Phys.} \textbf{13} 
(2006), no.~2, 158--187. arXiv:math.AP/0609649  
 
\bibitem{KKopS2011} 
A. Komech, E. Kopylova, H. Spohn, 
Scattering of solitons for Dirac equation coupled to a~particle, 
 \textit{J. Math. Analysis and Appl.} \textbf{383} (2011), no. 2, 265--290. 
arXiv:1012.3109 

\bibitem{IKV2011} 
V. Imaykin, A. Komech, B.Vainberg, 
Scattering of solitons for coupled wave-particle equations, 
\textit{J. Math. Analysis and Appl.} \textbf{389} 
(2012), no. 2, 713--740. 
arXiv:1006.2618 

\bibitem{KSK}
A. Komech, H. Spohn, M. Kunze, Long-time asymptotics for a classical particle
interacting with a scalar wave field,
{\em Comm. in PDEs} {\bf 22} (1997), no. 1-2, 307-335.

\bibitem{3w}
E. Kopylova,
Weighted energy decay for 3D wave equation,
{\em Asymptotic Anal.} {\bf 65} (2009), no. 1-2, 1-16.
 
\bibitem{Li}
 J.L.Lions,  ``Probl\`emes aux Limites dans les \'Equations aux
 D\'eriv\'ees Partielles'', Presses de l'Univ. de Montr\'eal, Montr\'eal, 1962.

\bibitem{SW} 
A. Soffer, M.I. Weinstein, 
Resonances, radiation damping and instability in Hamiltonian nonlinear 
wave equations,  {\em Invent. Math.} {\bf  136}   (1999), 
no. 1, 9-74. 
 
\end{thebibliography}
\end{document}